\documentclass[conference]{IEEEtran}
\IEEEoverridecommandlockouts
\usepackage{cite}
\usepackage{amsmath,amssymb,amsfonts,amsthm}
\usepackage{graphicx}
\usepackage{textcomp}
\usepackage{xcolor}
\usepackage{booktabs}
\usepackage{tikz}
\usetikzlibrary{arrows.meta,calc,decorations.pathreplacing,positioning}
\definecolor{copenblue}{RGB}{0,102,204}
\definecolor{copendark}{RGB}{0,51,102}
\definecolor{copenred}{RGB}{204,51,51}
\newtheorem{theorem}{Theorem}
\newtheorem{lemma}{Lemma}
\newtheorem{corollary}{Corollary}
\newtheorem{definition}{Definition}

\DeclareMathOperator{\TV}{TV}
\newcommand{\E}{\mathbb{E}}
\newcommand{\calX}{\mathcal{X}}
\newcommand{\calD}{\mathcal{D}}
\newcommand{\calA}{\mathcal{A}}
\newcommand{\dovr}{\bar\Delta_i^{\mathrm{ovr}}}

\begin{document}

\title{Inevitability of Encrypted Traffic Side-Channel Leakage in the Multi-Class Setting}

\author{
\IEEEauthorblockN{Guangjie Liu\textsuperscript{1}, Guang Cheng\textsuperscript{2}, Weiwei Liu\textsuperscript{3}}
\IEEEauthorblockA{\textsuperscript{1}School of Electronic \& Information Eng., Nanjing University of Information Science and Technology, Nanjing, China\\
\textsuperscript{2}School of Cyber Science \& Eng., Southeast University, Nanjing, China\\
\textsuperscript{3}School of Automation, Nanjing University of Science and Technology, Nanjing, China}
}

\maketitle

\begin{abstract}
The Side-Channel Existence Theorem proves $I(X;Y)>0$ in the binary, undefended setting, but is confined to pairwise arguments and ignores active defenses. We extend it to $k$ classes via the per-class decomposition $I(X;Y)=\sum_i\pi_i D_{\mathrm{KL}}(P_{Y|i}\|P_Y)$, with defense cost modelled by per-class Wasserstein-1 constraints $\sup_x W_1(Q_x^D,P_x)\le B$. Three results follow: (1) a summation-form MI lower bound over all active classes; (2) a cascade critical cost theorem and a per-class budget corollary, nonzero where the uniform-budget bound vanishes; (3) an accuracy corollary $\mathrm{Acc}^*\ge 2^{I_0}/k>1/k$. On a 95-class website fingerprinting dataset the measured MI has a strictly positive $95\%$ confidence lower bound under every defense tested. Against the strongest pairwise baseline---a convex program over all $\binom{k}{2}$ triangle constraints, also $\Theta(1)$ in $k$ under the same non-vanishing-gap conditions---the summation form is only $1.45\times$ stronger, so the case for the per-class decomposition is structural: only it gives each class a critical cost and a cascade. FRONT's apparent $122\times$ gap is inflated mainly by threshold exclusion rather than the inequality chain: on the active classes it is $21\times$, within $1.4\times$ of the $15\times$ measured undefended. Measuring the chain's two steps separately bounds the collapse onto one Lipschitz statistic below by $28\times$, against a divergence step measured at $1.5\times$. Undefended OVR distinguishability predicts post-defense per-class leakage at Spearman $\rho=0.62$--$0.77$, the transfer the certification procedure relies on. The framework carries over unchanged to a 100-class QUIC/TCP pair.
\end{abstract}


\section{Introduction}

The Side-Channel Existence Theorem \cite{liu2026inevitability} establishes that in efficiency-first encrypted communication systems, the mutual information (MI) between the semantic variable $X$ and the observable $Y$ is strictly positive: $I(X;Y)\ge\frac{1}{2\ln 2}(\frac{\rho[\bar\Delta-2L_\varphi C]}{2})^2>0$, with $\bar\Delta$ a mean gap between the two semantics, $C$ a drift bound, $L_\varphi$ a Lipschitz constant and $\rho$ a retention ratio (all defined in Sec.~\ref{sec:model}). Two limitations are fundamental. First, the framework is binary ($k=2$)---a corollary extends it to ``there exists at least one distinguishable pair,'' still a pairwise argument that does not characterize the $k$-class \emph{classification} capability real attackers have. Second, it does not address active defenses such as padding and morphing, widely deployed in practice.

Multi-class classification is the realistic setting: website fingerprinting targets $k=100$--$1000$ sites \cite{sirinam2018deep,zhao2024finegrained}, application identification dozens of types \cite{shen2023machine}, video fingerprinting single titles \cite{hasselquist2024raisingbar}. Attacks scale with class count---CountMamba \cite{deng2025countmamba} generalizes over closed-world, open-world and defended settings---while defenses advance in step \cite{wang2026frugal,shen2024palette}. None carries an information-theoretic guarantee: attackers do not know how well classification can theoretically do, defenders do not know how low MI can be pushed.

Prior information-theoretic treatments all speak about a \emph{given} system: Li et al.\ \cite{li2018measuring} measured $I(F;W)=6.63$ bits over 100 Tor sites against a $\log_2 100\approx6.64$ ceiling; Cherubin \cite{cherubin2017bayes} certified one defense via Bayes error bounds; FRUGAL \cite{wang2026frugal} constructs defenses driving MI down; and closest to this work, \cite{liu2026ratedistortion} computes the exact minimum of $I(X;Y)$ in a stationary memoryless class at $W_1$ cost $D$. None supplies a converse over \emph{all} of $\calD(B)$: a positive lower bound for every defense meeting the cost constraint. That is what this paper adds---the missing half of a sandwich characterization whose achievability side is \cite{liu2026ratedistortion,wang2026frugal}.

The key insight is to exploit the per-class MI decomposition $I(X;Y)=\sum_{i=1}^k\pi_i D_{\mathrm{KL}}(P_{Y|i}\|P_Y)$, where each term captures whether class $i$ can be \emph{identified} from the population mixture. We introduce per-class one-vs-rest (OVR) distinguishability, model defense cost via Wasserstein-1 constraints, and establish three results:

\begin{enumerate}
\item \textbf{$k$-Class Inevitability Theorem} (Theorem~\ref{thm:main}): a summation-form MI lower bound over all active classes, each contributing independently according to its OVR distinguishability---$1.45\times$ stronger than the best pairwise construction, and unlike it a \emph{per-class} certificate (Sec.~\ref{sec:pairwise}).

\item \textbf{Cascade Critical Cost Theorem} (Theorem~\ref{thm:cascade}): the ordered per-class critical costs are phase-transition thresholds. Its per-class budget corollary (Corollary~\ref{cor:perclass}) replaces the uniform budget by per-class costs $B_i$ and the prior-weighted average $\bar B_{\neg i}$, giving nonzero bounds where the uniform-budget bound vanishes.

\item \textbf{Attack Accuracy Corollary} (Corollary~\ref{cor:acc}): $\mathrm{Acc}^*\ge 2^{I_0}/k>1/k$ for any valid bound $I_0$, with a prior-free per-class balanced-accuracy counterpart for skewed priors.
\end{enumerate}

\textbf{Scope.} The core advance is a shift from ``pairwise distinction'' to ``per-class identification from the mixture,'' under a $W_1$ defense budget. We do not replace attack-specific accuracy analyses (e.g.\ Deep Fingerprinting \cite{sirinam2018deep}, over $98\%$ accurate undefended) or derive the optimal defense within $\calD(B)$. Lemma~\ref{lem:drift}, Theorem~\ref{thm:main} and Corollary~\ref{cor:perclass} are proved in full, Theorem~\ref{thm:cascade} and Corollary~\ref{cor:acc} in sketch. Code and results accompany the paper.

\section{System Model}\label{sec:model}

\subsection{$k$-Class Causal Chain with Defense}

Let the semantic space be $\calX=\{1,\ldots,k\}$ ($k\ge 2$) with prior $\boldsymbol{\pi}=(\pi_1,\ldots,\pi_k)$, $\pi_i>0$ for all $i$. Extending the causal chain of \cite{liu2026inevitability} to the defended $k$-class setting:
\begin{equation}\label{eq:chain}
X \xrightarrow{\ \mathcal{G}_A\ } \Xi_A \xrightarrow{\ \Pi\ } \Xi_P \xrightarrow{\ \Phi_D\ } \Xi_C^D \xrightarrow{\ N\ } \Xi_N^D \xrightarrow{\ \Theta\ } Y_D,
\end{equation}
where the defense $\Phi_D$ is internalized in the encryption layer. Write $z_P,z_C^D,z_N$ for realizations of $\Xi_P,\Xi_C^D,\Xi_N^D$, all in a common trajectory space with metric $d$, and $\varphi_i$ for the statistic reading class $i$---a \emph{family}, since what separates one class from the rest need not separate another; a single $\varphi$ is the case $\varphi_i\equiv\varphi$. The mappings inherit the properties of \cite{liu2026inevitability}, with mapping non-degeneracy now attached to the segment \emph{after} the defense---$\E[d(z_C^D,z_N)\mid X]\le C$, so that $C$ bounds the encryption and transport drift alone and the defense is accounted for once, by its own budget---together with Lipschitz robustness (every $\varphi_i$ is $L_\varphi$-Lipschitz on the trajectory space) and observation non-degeneracy (the observation map retains a fraction $\rho\in(0,1]$ of any mean gap, and every induced observation-layer statistic obeys $\|\psi_i\|_\infty\le M$). The constants $L_\varphi,\rho,M$ are \emph{uniform over the family}, which is what licenses summing per-class bounds from different $\varphi_i$ into one inequality. \cite{liu2026inevitability} takes $M=1$; here $M$ is measured from the observable's support. The defense family is defined via per-class Wasserstein-1 constraints:
\begin{equation}\label{eq:defense-set}
\calD(B)\;:=\;\big\{\Phi_D:\sup_{x\in\calX}W_1(Q_x^D,P_x)\le B\big\},
\end{equation}
where $P_x$ and $Q_x^D$ are the undefended and defended protocol-layer distributions for class $x$, and $B\ge 0$ is the defense budget. At $B=0$ every $\Phi_D\in\calD(0)$ leaves each class-conditional law unchanged ($Q_x^D=P_x$), and the model reduces to the undefended baseline.

$W_1$ is the right cost, and the obvious alternatives fail. By Kantorovich--Rubinstein duality an $L_\varphi$-Lipschitz statistic obeys $|\E_P[\varphi]-\E_Q[\varphi]|\le L_\varphi W_1(P,Q)$, so the constraint bounds exactly the quantity entering the proof chain. KL cannot serve as a budget---padding and morphing move mass onto lengths and timings absent from the undefended support, so $D_{\mathrm{KL}}(Q_x^D\|P_x)=\infty$ for exactly the defenses of interest---and TV saturates at $1$ once supports are disjoint, unable to tell ten bytes per packet from ten kilobytes. $W_1$ is the minimum expected transport cost of turning $P_x$ into $Q_x^D$, linear in the operator's overhead \emph{when the ground metric is that cost}---which our experiments do not instantiate: they use the Euclidean metric on the normalized feature space, so the $\hat B_i$ below are feature-space transport distances, not bandwidth or delay percentages.

\subsection{Per-Class One-vs-Rest Distinguishability}

\begin{definition}[Per-class OVR distinguishability]\label{def:ovr}
For each class $i\in\calX$, define its one-vs-rest distinguishability as
\begin{equation}\label{eq:delta-ovr}
\dovr\;:=\;\big|\E[\varphi_i(z_P)\mid X=i]-\E[\varphi_i(z_P)\mid X\neq i]\big|,
\end{equation}
where $\E[\varphi_i(z_P)\mid X\neq i]=\sum_{j\neq i}\frac{\pi_j}{1-\pi_i}\E[\varphi_i(z_P)\mid X=j]$ is the prior-weighted conditional expectation over non-$i$ classes.
\end{definition}

$\dovr$ measures how far class $i$ deviates from the mixture of all others, the natural quantity for $k$-class identification.

Why OVR rather than pairwise is a structural question. The decomposition $I(X;Y)=\sum_i\pi_i D_{\mathrm{KL}}(P_{Y|i}\|P_Y)$ is indexed by \emph{classes}, each summand a distance to the \emph{mixture} $P_Y$. A pairwise gap $\bar\Delta_{ij}$ constrains neither distance directly; the only step available is the triangle inequality $\TV(P_{Y|i},P_Y)+\TV(P_{Y|j},P_Y)\ge\TV(P_{Y|i},P_{Y|j})$. That route is stronger than it looks---all $\binom{k}{2}$ constraints may be imposed at once, so under the same conditions that make the summation form $\Theta(1)$, a convex program over them is too (Sec.~\ref{sec:pairwise}). What it cannot yield is a \emph{per-class} statement: such constraints certify that a set of pairs jointly leaks, never that class $i$ is identifiable, so they support no critical cost $B_i^*$ and no cascade---the structure Secs.~\ref{sec:main} and \ref{sec:disc} rest on.

Writing $\mu_j:=\E[\varphi_i(z_P)|X=j]$ and $w_j:=\pi_j/(1-\pi_i)$ gives $\dovr=|\sum_{j\neq i}w_j(\mu_i-\mu_j)|$, a \emph{signed} average of pairwise distances. For a ``centroid'' class, whose $\mu_i$ is the weighted mean of all others, the terms cancel and $\dovr=0$ even though the class is pairwise distinguishable from every other, so the sufficient condition says nothing about it---a cancellation invisible in the binary theory, where OVR reduces to $\bar\Delta_{12}$, and confronted with data in Sec.~\ref{sec:disc}.

\subsection{Critical Costs and Active Class Set}

\begin{definition}[Per-class critical cost and active class set]\label{def:critical}
For class $i$ with $\dovr>2L_\varphi C$, define
\begin{equation}\label{eq:Bi-star}
B^*_i\;:=\;\frac{\dovr-2L_\varphi C}{2L_\varphi}.
\end{equation}
The active class set at budget $B$ is $\calA^{\mathrm{cls}}(B):=\{i\in\calX:\dovr>2L_\varphi(C+B)\}$, with effective class count $k_{\mathrm{eff}}(B):=|\calA^{\mathrm{cls}}(B)|$.
\end{definition}

Ordering the critical costs as $B^*_{(1)}\le\cdots\le B^*_{(k')}$ ($k'=k_{\mathrm{eff}}(0)$), the active set falls in a cascade as $B$ grows---one class exiting at each distinct threshold, several at a tied one---until the global critical cost $B^*_{\max}=B^*_{(k')}$ extinguishes every lower bound. Fig.~\ref{fig:main}(a) shows the cascade measured on real traffic.

\section{Main Results}\label{sec:main}

\subsection{$k$-Class Inevitability Theorem}

The whole argument rests on one lemma, which converts the $W_1$ budget into a guaranteed separation at the observation layer.

\begin{lemma}[OVR drift under a $W_1$ budget]\label{lem:drift}
Under mapping non-degeneracy $C$, Lipschitz robustness $L_\varphi$ and observation non-degeneracy $\rho$, for any $\Phi_D\in\calD(B)$ and any class $i$,
\begin{equation}\label{eq:lemma}
\TV\big(P_{Y_D|i},P_{Y_D|X\neq i}\big)\ \ge\ \frac{\rho\,\delta_i(B)}{2M},
\end{equation}
where $\delta_i(B):=\dovr-2L_\varphi(C+B)$ and $M$ bounds the induced observation-layer statistic, $\|\psi_i\|_\infty\le M$. A $1$-Lipschitz $\psi$ on a support of diameter $D$ admits $M=D/2$ after centring, so $M$ is measurable rather than assumed; \cite{liu2026inevitability} takes $M=1$.
\end{lemma}

\begin{proof}
Fix $i$ and let $\varphi_i$ realize \eqref{eq:delta-ovr}; it is $L_\varphi$-Lipschitz by hypothesis. On the class-$i$ branch the drift of $\E[\varphi_i]$ from $\Xi_P$ to $\Xi_N^D$ splits along the two consecutive segments of \eqref{eq:chain}: over $\Xi_P\!\to\!\Xi_C^D$ the defense moves the class-$i$ law by $W_1(Q_i^D,P_i)\le B$, contributing at most $L_\varphi B$ by duality, and over $\Xi_C^D\!\to\!\Xi_N^D$ mapping non-degeneracy contributes at most $L_\varphi C$. The segments are disjoint, so these add rather than double-count: $|\E[\varphi_i(z_N)|X{=}i]-\E[\varphi_i(z_P)|X{=}i]|\le L_\varphi(C+B)$. On the non-$i$ branch the conditional law is $\sum_{j\neq i}w_jP_{\cdot|j}$ with $w_j=\pi_j/(1-\pi_i)$ summing to $1$, and the same bound holds for every $j$, so it survives the averaging. Subtracting, the network-layer OVR gap is at least $\dovr-2L_\varphi(C+B)=\delta_i(B)$, of which observation non-degeneracy retains a fraction $\rho$. The induced statistic $\psi_i$ then carries a mean gap of at least $\rho\,\delta_i(B)$, and the bounded-statistic lemma of \cite{liu2026inevitability}---$\|f\|_\infty\le M$ and $|\E_Pf-\E_Qf|\ge\delta$ imply $\TV(P,Q)\ge\delta/(2M)$, by the dual form $\TV=\tfrac12\sup_{\|g\|_\infty\le1}|\E_Pg-\E_Qg|$---gives \eqref{eq:lemma}.
\end{proof}

\begin{theorem}[$k$-Class Inevitability]\label{thm:main}
Let $\calX=\{1,\ldots,k\}$ ($k\ge 2$) with $\pi_i>0$ for all $i$. Under the conditions of Lemma~\ref{lem:drift}, if $\calA^{\mathrm{cls}}(B)\neq\varnothing$, then for any $\Phi_D\in\calD(B)$, $I(X;Y_D)>0$ with
\begin{multline}\label{eq:main-bound}
I(X;Y_D)\ge\frac{2}{\ln 2}\!\sum_{i\in\calA^{\mathrm{cls}}(B)}\!\pi_i(1\!-\!\pi_i)^2
\left(\frac{\rho\,\delta_i(B)}{2M}\right)^{\!2}.
\end{multline}
\end{theorem}

\begin{proof}
Decompose $I(X;Y_D)=\sum_{i=1}^k\pi_iD_{\mathrm{KL}}(P_{Y_D|i}\|P_{Y_D})$, an identity. Fix an active class $i$. Writing the marginal as the two-component mixture $P_{Y_D}=\pi_iP_{Y_D|i}+(1-\pi_i)P_{Y_D|X\neq i}$ gives $P_{Y_D|i}-P_{Y_D}=(1-\pi_i)(P_{Y_D|i}-P_{Y_D|X\neq i})$ pointwise, hence the \emph{identity}
\begin{equation}\label{eq:mixture-id}
\TV(P_{Y_D|i},P_{Y_D})=(1-\pi_i)\,\TV(P_{Y_D|i},P_{Y_D|X\neq i}),
\end{equation}
no inequality having been used yet. Lemma~\ref{lem:drift} bounds the right-hand factor below by $\rho\delta_i(B)/(2M)$, and $\delta_i(B)>0$ precisely because $i\in\calA^{\mathrm{cls}}(B)$. Pinsker's inequality in bits, $D_{\mathrm{KL}}(P\|Q)\ge\frac{2}{\ln 2}\TV(P,Q)^2$, then yields
\begin{equation}\label{eq:perterm}
D_{\mathrm{KL}}(P_{Y_D|i}\|P_{Y_D})\ \ge\ \frac{2}{\ln 2}(1-\pi_i)^2\Big(\frac{\rho\,\delta_i(B)}{2M}\Big)^{\!2}.
\end{equation}
Multiplying by $\pi_i$ and summing over $i\in\calA^{\mathrm{cls}}(B)$ gives \eqref{eq:main-bound}: the discarded terms are the inactive ones, each a nonnegative $\pi_iD_{\mathrm{KL}}(\cdot\|\cdot)$, so dropping them preserves the inequality, and $\calA^{\mathrm{cls}}(B)\neq\varnothing$ leaves at least one strictly positive summand.
\end{proof}

The essential distinction from the binary theorem is the \emph{summation structure}: each active class contributes independently. Under a uniform prior with a common residual $\delta:=\bar\Delta^{\mathrm{ovr}}-2L_\varphi(C+B)$, the sum has $k$ terms of size $\frac{2}{k\ln 2}(1-1/k)^2[\rho\delta/2M]^2$ and increases to $\frac{2}{\ln 2}[\rho\delta/2M]^2=\Theta(1)$ as $k\to\infty$: the guarantee does not dilute as the problem grows. At $k=2$ the two summands combine as $\pi_1(1-\pi_1)^2+\pi_2(1-\pi_2)^2=\pi_1\pi_2$, so \eqref{eq:main-bound} becomes $\frac{2}{\ln2}\pi_1\pi_2[\rho\delta/2M]^2$, $\rho\delta/2M$ being the Lemma~\ref{lem:drift} floor on $\TV(P_{Y_D|1},P_{Y_D|2})$---structurally the binary bound $\frac{2}{\ln2}P(x)P(x')\TV^2$ of \cite{liu2026inevitability}, recovered for \emph{every} prior rather than only the balanced one. The content of Theorem~\ref{thm:main} therefore lies in the regime $k\ge3$.

\subsection{Cascade Critical Cost Theorem}

\begin{theorem}[Cascade Critical Cost]\label{thm:cascade}
Under the conditions of Theorem~\ref{thm:main}, let $k'=k_{\mathrm{eff}}(0)\ge 1$ and order the critical costs as $B^*_{(1)}\le\cdots\le B^*_{(k')}$. Then: (i)~$k_{\mathrm{eff}}(B)=k'-|\{m:B^*_{(m)}\le B\}|$; (ii)~$B^*_{\max}:=B^*_{(k')}$ is the minimum budget to extinguish all bounds; (iii)~the MI lower bound on each interval $[B^*_{(m)},B^*_{(m+1)})$ is a sum of $k'-m$ quadratic functions of $B$, joined $C^1$ at each threshold, where the curvature jumps.
\end{theorem}

\textit{Proof sketch of Theorem~\ref{thm:cascade}.} Part~(i) is Definition~\ref{def:critical}: class $(m)$ exits when $B\ge B^*_{(m)}$. Part~(ii) follows since for $B\ge B^*_{\max}$ the sum in \eqref{eq:main-bound} is empty. Part~(iii) holds because $\calA^{\mathrm{cls}}(B)$ is constant on each interval and every summand is quadratic in $B$, each vanishing quadratically at its own threshold, so the sum is continuously differentiable there and only its second derivative jumps. \hfill$\square$

\begin{corollary}[Per-class budget]\label{cor:perclass}
Let $B_i:=W_1(Q_i^D,P_i)$, $\bar B_{\neg i}:=\sum_{j\neq i}\frac{\pi_j}{1-\pi_i}B_j$, and write $\delta_i^{\mathrm{pc}}:=\dovr-L_\varphi(C+B_i)-L_\varphi(C+\bar B_{\neg i})$ for the per-class residual, $\calA^{\mathrm{pc}}:=\{i:\delta_i^{\mathrm{pc}}>0\}$ and $k^{\mathrm{pc}}_{\mathrm{eff}}:=|\calA^{\mathrm{pc}}|$. If $\calA^{\mathrm{pc}}\neq\varnothing$, then
\begin{multline}\label{eq:perclass-bound}
I(X;Y_D)\ge\frac{2}{\ln 2}\!\sum_{i\in\calA^{\mathrm{pc}}}\!\pi_i(1\!-\!\pi_i)^2\\
\times\left(\frac{\rho\,\delta_i^{\mathrm{pc}}}{2M}\right)^{\!2}.
\end{multline}
Theorem~\ref{thm:main} is the special case $B_i\equiv B$, so \eqref{eq:perclass-bound} is the operative form throughout; we write $I_0^{\mathrm{pc}}$ for its right-hand side.
\end{corollary}

\begin{proof}[Proof of Corollary~\ref{cor:perclass}]
Only the two drift bounds inside Lemma~\ref{lem:drift} change. On the class-$i$ branch the budget is $B_i$, so $|\E[\varphi_i(z_N)|X{=}i]-\E[\varphi_i(z_P)|X{=}i]|\le L_\varphi(C+B_i)$ exactly as before. On the non-$i$ branch the conditional law is $\sum_{j\neq i}w_jP_{\cdot|j}$ with $w_j=\pi_j/(1-\pi_i)$, so the drift of its mean is the $w$-average of the per-class drifts, and
\begin{equation}\label{eq:convex-step}
\big|\mu^N_{\neg i}-\mu^P_{\neg i}\big|\le\sum_{j\neq i}w_j L_\varphi(C+B_j)=L_\varphi(C+\bar B_{\neg i}),
\end{equation}
using $\sum_{j\neq i}w_j=1$ and the definition of $\bar B_{\neg i}$. This is where the corollary gains: the mixture side is driven by the prior-weighted \emph{average} of the other classes' costs, not their maximum. Subtracting the two drifts replaces the symmetric threshold $2L_\varphi(C+B)$ by $L_\varphi(C+B_i)+L_\varphi(C+\bar B_{\neg i})$; the rest of the proof of Theorem~\ref{thm:main} applies verbatim with $\calA^{\mathrm{pc}}$ in place of $\calA^{\mathrm{cls}}(B)$.
\end{proof}

Consequently the per-class threshold is never larger than the uniform-budget threshold $2L_\varphi(C+\max_jB_j)$, and is strictly smaller as soon as the costs are heterogeneous---the norm in practice, since padding defenses impose different per-class costs depending on each site's traffic profile.

\subsection{Attack Accuracy Corollary}

\begin{corollary}[Attack accuracy bound]\label{cor:acc}
Let $I_0$ be any valid lower bound on $I(X;Y_D)$---in particular the right-hand side of \eqref{eq:main-bound} or, for heterogeneous costs, $I_0^{\mathrm{pc}}$ of \eqref{eq:perclass-bound}. The Bayes-optimal $k$-class accuracy satisfies
\begin{equation}\label{eq:acc-bound}
\mathrm{Acc}^*\;\ge\;\max\Big\{\max_i\pi_i,\ \frac{2^{I_0}}{k}\Big\},
\end{equation}
which exceeds $1/k$ whenever $I_0>0$, i.e.\ whenever some class remains active.
\end{corollary}

\textit{Proof sketch.} From $I(X;Y_D)\ge I_0$, $H(X|Y_D)\le\log_2 k-I_0$; the posterior-maximum inequality $\max_iP(i|y)\ge 2^{-H(X|Y_D=y)}$ \cite{cover2006elements} with Jensen's inequality gives $\mathrm{Acc}^*\ge 2^{-H(X|Y_D)}\ge 2^{I_0}/k$. The term $\max_i\pi_i$ is the trivial majority-class guess. \hfill$\square$

The factor $2^{I_0}>1$ is the multiplicative improvement over blind guessing. The $\max$ in \eqref{eq:acc-bound} matters: under a skewed prior---Zipf-like page popularity---$\max_i\pi_i$ dominates $2^{I_0}/k$ and the bound says nothing about \emph{which} class is identified. There the operational statement is the per-class guarantee Lemma~\ref{lem:drift} yields directly. The optimal test of ``$X=i$'' against ``$X\neq i$'' attains balanced accuracy---the mean of its true-positive and true-negative rates---at least
\begin{equation}\label{eq:acc-ovr}
\mathrm{Acc}^{\mathrm{ovr}}_i\;\ge\;\tfrac{1}{2}\big(1+\TV(P_{Y_D|i},P_{Y_D|X\neq i})\big)\;\ge\;\tfrac{1}{2}\Big(1+\frac{\rho\,\delta_i}{2M}\Big),
\end{equation}
the first step the standard identity for the Bayes error of a two-hypothesis test under equal weighting, the second Lemma~\ref{lem:drift} with $\delta_i=\delta_i(B)$, or its per-class counterpart $\delta_i=\delta_i^{\mathrm{pc}}$ under heterogeneous realized costs---the version instantiated in Sec.~\ref{sec:main-verif}. Balanced accuracy is prior-free by construction, so \eqref{eq:acc-ovr} survives prior skew, and it is the question a defender of one high-value site actually asks. The counts $k_{\mathrm{eff}}(B)$ and $k^{\mathrm{pc}}_{\mathrm{eff}}$ are exactly how many classes make \eqref{eq:acc-ovr} nontrivial, under a uniform budget and under realized per-class costs; both fall to zero as the defense grows.

\begin{figure*}[t]
\centering
\includegraphics[width=\textwidth]{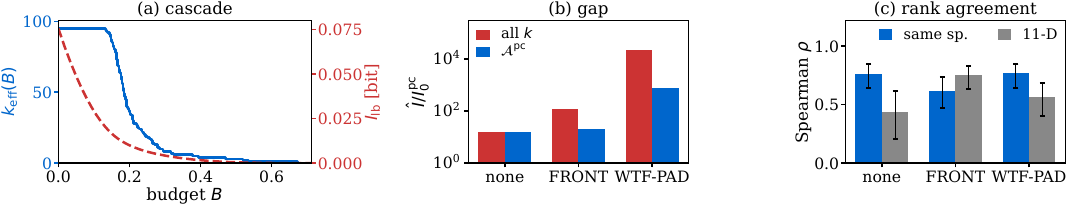}
\caption{Cascade, gap decomposition and rank agreement on CW}\label{fig:main}
\end{figure*}

\section{Empirical Validation}\label{sec:empirical}

\subsection{Protocol}\label{sec:protocol}

Experiments use the closed-world dataset of \cite{sirinam2018deep} (95 sites, 200 traces each, 19\,000 total) undefended and under three defenses---FRONT \cite{gong2020front} (front-loaded dummy injection), WTF-PAD \cite{juarez2016wtfpad} (adaptive padding), TrafficSliver \cite{cadena2020trafficsliver} (multipath splitting)---plus the QUIC/TCP pair of \cite{smith2021quic} (100 sites, 100 traces each, 10\,000 per stack), testing whether the framework survives a change of protocol stack. Each trace is a signed-time packet sequence $\{(\pm 1,t_j)\}$ yielding a 3-D observable $Y^{(3)}$ (inbound ratio, mean IAT, direction flip rate) for MI, and an 11-D $Y^{(11)}$ (those three plus packet count, flow duration, IAT standard deviation, IAT 25th and 75th percentiles, burst count, mean burst length, cumulative direction at the 30\% position) for $W_1$ quantities. Coordinates are mapped into $[0,1]$ by the undefended 99th percentiles (direction ratios already lie there; cumulative direction is rescaled affinely). $Y^{(3)}$ is exactly the first three coordinates of $Y^{(11)}$, so the data-processing inequality invoked below applies.

The measured MI, $\hat I_{\mathrm{MM}}$, is the plug-in estimator with Miller--Madow correction on $Y^{(3)}$ (5 bins per dimension, $5^3=125$ symbols on each dataset's own range; 200 stratified bootstrap replicates, percentile 95\% intervals bias-shifted). Distinguishability and cost are exact earth-mover distances on $Y^{(11)}$: $\hat{\bar\Delta}_i^{\mathrm{ovr}}=W_1(\hat P_{Y^{(11)}|i},\hat P_{Y^{(11)}|\neg i})$, the $\neg i$ mixture subsampled to 2\,000 traces, and $\hat B_i$ the defended-to-baseline $W_1$ per class. By duality $W_1$ selects the optimal $1$-Lipschitz statistic, so $\hat{\bar\Delta}_i^{\mathrm{ovr}}$ is the tightest instantiation of Definition~\ref{def:ovr} the sample allows. The optimizer differs by class---no single $\varphi$ realizes all 95---which is why Definition~\ref{def:ovr} ranges over a family $\{\varphi_i\}$ with constants uniform across it.

Walkie-Talkie \cite{wang2017walkietalkie} was measured but excluded from Table~\ref{tab:main}: its half-duplex traces fill only 33 of the 125 symbols and the estimate is unstable (CI $[0.67,1.37]$).

The constant $M$ of Lemma~\ref{lem:drift} is measured, not assumed, by one rule in whichever space the $W_1$ quantities are estimated---half the exact diameter of the support entering them: $2.8323/2=1.4161$ in $Y^{(11)}$, $1.3344/2=0.667$ in $Y^{(3)}$. Each space needs its own, since $M$ must dominate the statistic \emph{in that space} for \eqref{eq:lemma} to hold: importing the $Y^{(11)}$ constant into $Y^{(3)}$ loosens the bound there by $4.5\times$, and taking $M=1$ as the binary theorem does would inflate the $Y^{(11)}$ bound by $2.00\times$ and deflate the $Y^{(3)}$ one by $2.25\times$. Proxy parameters are fixed at $C=0$, $L_\varphi=1$, $\rho=1$. Both $C=0$ and $\rho=1$ \emph{maximize} the bound, making the empirical test strictest---but not conservative in the sense of guaranteed validity, which would need $C$ above and $\rho$ below their true values, and the plug-in estimates carry no one-sided finite-sample control. Table~\ref{tab:main} is a stress test, not a certificate. $L_\varphi=1$ is not conservative but self-consistent: a statistic realizing a $W_1$ distance is $1$-Lipschitz, and since $\hat{\bar\Delta}_i^{\mathrm{ovr}}$ and $\hat B_i$ are both $W_1$ in the same normalized space, the residual is scale-consistent only at $L_\varphi=1$.

One asymmetry governs how Table~\ref{tab:main} may be read. Its bound uses the 11-D $\dovr$, so it bounds $I(X;Y^{(11)})$, while the measured column is $\hat I_{\mathrm{MM}}$ on $Y^{(3)}$ and $I(X;Y^{(3)})\le I(X;Y^{(11)})$: the verification is \emph{sufficient but not necessary}, clearing a bar set for a richer observable with a poorer one. The last column is correspondingly \emph{not} a tightness measure; the like-for-like comparison, both sides in $Y^{(3)}$, is in Sec.~\ref{sec:gap} and is $2.8\times$ less favourable.

\subsection{Main Verification}\label{sec:main-verif}

Table~\ref{tab:main} reports the measured MI against Corollary~\ref{cor:perclass}. Undefended, the $W_1$ OVR spectrum spans $0.261$--$1.343$ (mean $0.429$), all 95 classes are active, and the ordered critical costs run from $\hat B^*_{(1)}=0.130$ to $\hat B^*_{\max}=0.671$---a $5.1\times$ spread, the cascade of Fig.~\ref{fig:main}(a). Every defense drives $\max_i\hat B_i$ above $\hat B^*_{\max}$, so Theorem~\ref{thm:main}'s uniform-budget bound vanishes in all three defended settings; the per-class budget still keeps 7 classes active under FRONT and 1 under WTF-PAD---the difference between a vacuous and a nonvacuous bound. The surviving mass is concentrated: one class ($\hat{\bar\Delta}^{\mathrm{ovr}}=1.343$, $\hat B_i=0.169$, $\delta^{\mathrm{pc}}=0.789$) supplies $63\%$ of the bound, while the two weakest survivors, clearing their thresholds by $0.011$ and $0.012$, contribute $0.03\%$ between them---the summand is quadratic in $\delta^{\mathrm{pc}}$, so the ranking follows the residual, not raw distinguishability. Restated per class through \eqref{eq:acc-ovr}: even under FRONT, deciding whether a trace belongs to the leading class has \emph{balanced} accuracy $\ge0.639$ against $0.5$ for a coin, $0.737$ undefended. Corollary~\ref{cor:acc}'s $k$-class statement is far weaker---$\mathrm{Acc}^*\ge0.0111$ against $1/k=0.0105$ undefended---since $2^{I_0}\approx1$ whenever $I_0\ll1$ bit. A total of $3.6\times10^{-3}$ bit across a 95-class problem is negligible; the same bound read class by class is not.

Every measured MI has a strictly positive 95\% CI lower bound, even where the theory gives nothing. The ordering is informative: FRONT leaves $0.445$ bit at mean cost $0.383$, WTF-PAD $0.585$ at $0.709$, TrafficSliver $0.482$ at $0.791$---the cheapest defense here is also the least leaky, and what TrafficSliver's extra spend buys over FRONT is not a lower MI but the collapse of the certifiable set. Finally, the QUIC/TCP pair behaves like the CW baseline: all 100 classes active, gaps of $8.3\times$--$8.9\times$, bounds agreeing to $0.03\%$ despite QUIC's longer tail ($\max_i\hat{\bar\Delta}_i^{\mathrm{ovr}}=1.715$ against $1.629$) and lower floor ($0.359$ against $0.371$)---the sum over 100 classes averages such differences out. Each stack is normalized by its own percentiles, so the comparison spans two scales.

\begin{table}[tb]
\caption{Measured MI versus theoretical lower bounds}\label{tab:main}
\begin{center}
\setlength{\tabcolsep}{2.6pt}
\begin{tabular}{@{}lccccc@{}}
\toprule
Setting & $\max_i\hat B_i$ & $k^{\mathrm{pc}}_{\mathrm{eff}}$ & $I_0^{\mathrm{pc}}$ & $\hat I_{\mathrm{MM}}$ [95\% CI] & gap \\
\midrule
\multicolumn{6}{l}{\emph{CW, $k=95$}}\\
none          & --    & 95 & $7.55\!\times\!10^{-2}$ & 1.136 [1.108, 1.169] & $15.0$ \\
FRONT         & 1.201 & 7  & $3.64\!\times\!10^{-3}$ & 0.445 [0.428, 0.460] & $122$ \\
WTF-PAD       & 1.095 & 1  & $2.55\!\times\!10^{-5}$ & 0.585 [0.568, 0.604] & $2.3\!\times\!10^{4}$ \\
TrafficSliver & 1.040 & 0  & $0$                     & 0.482 [0.439, 0.543] & -- \\
\midrule
\multicolumn{6}{l}{\emph{Smith, $k=100$, undefended}}\\
TCP           & --    & 100 & $1.29\!\times\!10^{-1}$ & 1.152 [1.114, 1.235] & $8.9$ \\
QUIC          & --    & 100 & $1.29\!\times\!10^{-1}$ & 1.073 [1.030, 1.140] & $8.3$ \\
\bottomrule
\end{tabular}
\end{center}
\vspace{-2pt}
\footnotesize $k^{\mathrm{pc}}_{\mathrm{eff}}$, $I_0^{\mathrm{pc}}$ are Corollary~\ref{cor:perclass}'s per-class quantities; Theorem~\ref{thm:main}'s uniform-budget counterparts are zero in all three defended settings. ``gap'' is $\hat I_{\mathrm{MM}}/I_0^{\mathrm{pc}}$ over all $k$ classes---see Sec.~\ref{sec:protocol} for how to read it and Sec.~\ref{sec:gap} for the matched comparison.
\end{table}

\subsection{Where the Gap Comes From}\label{sec:gap}

The bound is conservative, and the gap varies by three orders of magnitude across Table~\ref{tab:main}---not noise. The ratio factors into three effects: an index-set mismatch, the inequality chain, and the choice of estimation space. The first is inactive undefended, leaving $15.0=42.3/2.81$; under FRONT it splits $122$ into $5.84\times20.9$. Fig.~\ref{fig:main}(b) plots matched against full-set ratios.

\emph{Index-set mismatch} comes first. The bound sums only over $\calA^{\mathrm{pc}}$ while $\hat I_{\mathrm{MM}}$ measures all $k$ classes; since the measurement decomposes as the bound does, up to the global Miller--Madow correction (split in proportion to the uncorrected shares), we can restrict it to the same index set. Under FRONT the seven active classes carry $7.4\%$ of the prior mass and $\hat I_{\calA}=0.0762$ bit (95\% CI $[0.0720,0.0807]$), so the matched gap is $20.9\times$, not $122\times$; undefended, where nothing is excluded, it is $15.0\times$. On the same classes FRONT's gap is within $1.4\times$ of the no-defense gap: the apparent order-of-magnitude degradation under defense is an artifact of counting classes the theorem never claimed.

The second source is the \emph{inequality chain}, measured with bound and measurement in the same 3-D observable and with $M$ measured in that space: there the bound is $2.68\times10^{-2}$ bit against $1.136$ measured, a factor $42.3$, and both steps are separately measurable. Applying Pinsker to the \emph{measured} $\TV(P_{Y|i},P_{Y})$ instead of the Lemma~\ref{lem:drift} floor gives $0.756$ bit, so the divergence inequality accounts for $1.50\times$ (a ratio of estimates, not exact divergences---the numerator carries the Miller--Madow correction; $1.54$ without it) and the remaining $28.2\times$ is the step replacing a total variation distance by the mean gap of a \emph{single} $1$-Lipschitz statistic. One caveat is structural: the measured side uses the 125-symbol discretization while the floor is computed on the continuous $Y^{(3)}$. Data processing signs this only partly: $I_c\ge\hat I$, $\TV_c\ge\widehat\TV$, so $42.3$ and $28.2$ are underestimates, while the Pinsker factor---a ratio of two quantities that each grow---is unsigned.

The \emph{estimation space} runs the other way: estimating in $Y^{(11)}$ rather than $Y^{(3)}$ raises the bound by $2.8\times$ net (the richer space lifts $\hat{\bar\Delta}$ by $12.7\times$ in the squared term but admits a $4.5\times$ larger $M^2$), turning the same-space $42.3$ into Table~\ref{tab:main}'s $15.0$. As Sec.~\ref{sec:protocol} notes, that buys a valid bound on a richer observable, not a tighter one on the observable measured. The asymmetry is large and not the one we expected: at least $28\times$ sits in collapsing a multidimensional difference onto one scalar, against a divergence step measured at $1.50\times$---a quantized figure with no proven direction. Only the $28\times$ bounds anything; Sec.~\ref{sec:disc} acts on it.

WTF-PAD shows the mechanism cleanly. Exactly one class survives, clearing its threshold by $0.083$ out of $\dovr=1.343$; the summand being quadratic in that residual, the bound is $2.55\times10^{-5}$ bit---near zero by construction, and the matched gap correspondingly $757\times$. TrafficSliver, with $\hat B_i$ exceeding every $\hat B^*_i$, admits no active class and certifies nothing, while the measurement still gives $0.482$ bit with a strictly positive CI. This is the honest boundary of the sufficient condition: when the per-class cost approaches the OVR spectrum itself, the framework goes silent well before the leakage does.

\subsection{What the OVR Decomposition Buys}\label{sec:pairwise}

The comparison of Sec.~\ref{sec:model} can be priced, giving the pairwise route its strongest form. Writing $x_i=\TV(P_{Y|i},P_Y)$, the true distances satisfy $x_i+x_j\ge t_{ij}$ for \emph{every} pair, $t_{ij}$ the floor Lemma~\ref{lem:drift} gives on the pair $(i,j)$ rather than $(i,\neg i)$---which needs $L_\varphi,\rho,M$ uniform over a per-\emph{pair} family $\{\varphi_{ij}\}$, a stronger hypothesis than our theorem uses, granted here to the baseline. Since $I\ge\frac{2}{\ln2}\sum_i\pi_ix_i^2$ counts each class once, all $\binom{k}{2}$ constraints may be imposed at once without double counting, so the strongest pairwise bound is the convex program $\min_{x\ge0}\frac{2}{\ln2}\sum_i\pi_ix_i^2$ subject to them. On CW it returns $5.21\times10^{-2}$ bit, against $4.47\times10^{-2}$ for a disjoint matching (which discards constraints sharing a class), $6.35\times10^{-3}$ for the single best pair, and $7.55\times10^{-2}$ for the OVR sum: the summation form is $1.45\times$ stronger, not an order of magnitude.

Two things follow. First, the case for OVR is structural: the program names no class and supports neither the cascade nor the procedure of Sec.~\ref{sec:disc}. Second, adding the OVR floors $x_i\ge(1-\pi_i)\rho\delta_i(0)/2M$ to that program returns exactly the OVR value---every pairwise constraint is slack at the OVR optimum by at least $0.028$ in TV---so the pairwise family contributes nothing the per-class decomposition has not already supplied.

\subsection{Does the Bound Rank Classes Correctly?}\label{sec:pred}

Strict positivity is a weak test of a theorem whose content is a \emph{per-class} decomposition; the sharper question is whether $\dovr$ predicts \emph{which} classes leak. Measurement and bound are both indexed by class, so the two vectors can be correlated---provided both live in one feature space. Undefended distinguishability $\hat{\bar\Delta}_i^{\mathrm{ovr}}$---what step~(i) of Sec.~\ref{sec:disc} sorts on---predicts post-defense leakage well: in the same 3-D observable as the MI, Spearman $\rho=0.76$ undefended, $0.62$ under FRONT, $0.77$ under WTF-PAD, 95\% intervals $[0.64,0.85]$, $[0.47,0.74]$, $[0.65,0.84]$ from a bootstrap over the 95 \emph{classes}---between-class variability of the agreement, not within-class estimation error in $\hat{\bar\Delta}_i^{\mathrm{ovr}}$ or the contributions. The per-class \emph{bound term} $\delta_i^{\mathrm{pc}}$ subtracts realized costs, so undefended it \emph{is} $\hat{\bar\Delta}_i^{\mathrm{ovr}}$; under a defense, in the same 3-D space, it is not merely weaker but nearly degenerate---$\rho=0.25$ under FRONT, where only two classes stay above the per-class threshold, undefined under WTF-PAD, where none does. So the bound does more than assert $I>0$, but what transfers across a defense is the undefended ordering, not the defended summand---fortunate for the procedure, and a caution against reading the summand as a leakage predictor.

Substituting Table~\ref{tab:main}'s 11-D estimate while keeping the 3-D measurement perturbs the agreement unsystematically---$0.76\!\to\!0.44$, $0.62\!\to\!0.75$, $0.77\!\to\!0.56$, the two series of Fig.~\ref{fig:main}(c)---because the spaces agree only at $\rho=0.55$. That is mismatch, not a property of any defense: the procedure of Sec.~\ref{sec:disc} is only as good as the ordering it assumes, which should be estimated in the space the adversary is expected to observe.

\subsection{Robustness}\label{sec:robust}

Varying the bin count $k_b\in\{3,\ldots,7\}$ on CW gives $\hat I_{\mathrm{MM}}\in[0.571,1.283]$ bit, strictly positive throughout; the Miller--Madow correction is $0.025$ bit, $2.1\%$ of the estimate, so the positivity conclusion does not rest on the bias correction. Bin edges follow each dataset's own range; re-measuring on the \emph{undefended} edges---one fixed quantizer---gives $0.445\!\to\!0.438$, $0.585\!\to\!0.575$, $0.482\!\to\!0.553$, no sign change and FRONT still lowest. Varying the $\neg i$ subsample over $\{1000,2000,4000\}$ changes $I_0$ by under $0.2\%$ ($0.07549$--$0.07560$ bit), leaving $k_{\mathrm{eff}}=95$. Replacing the $W_1$ estimate of $\dovr$ by the centroid distance in the same 11-D space---one fixed linear statistic instead of the optimal Lipschitz one---costs a factor $1.56$, quantifying the value of the duality-based estimator. No choice here affects a qualitative conclusion.

\section{Discussion}\label{sec:disc}

\textbf{From bound to defense certification.} The results give a procedure needing no attack model. (i)~Estimate $\{\hat{\bar\Delta}_i^{\mathrm{ovr}}\}$ undefended and sort the induced $\{\hat B^*_i\}$. (ii)~For a target $k_0$---``at most $k_0$ classes may remain identifiable''---read $B=B^*_{(k'-k_0)}$ off the staircase of Fig.~\ref{fig:main}(a), $B=0$ if $k_0\ge k'$. It is a floor, not a recipe: spending less provably misses the target, more does not guarantee reaching it. (iii)~After deployment, measure $\hat B_i$ and evaluate Corollary~\ref{cor:perclass}: if $|\calA^{\mathrm{pc}}|>k_0$ the defense fails the target, against an adversary reading the feature map the estimates were made in---no classifier need be trained. That is a converse's operational value: it \emph{falsifies} defenses rather than certifying them. The verdict is only as sound as the instantiation---with $C=0$, $\rho=1$ and plug-in $W_1$ estimates it is an empirical criterion, not a proof.

On our numbers $k_0=10$ needs $\hat B^*_{(85)}=0.284$, $k_0=1$ needs $0.535$: going from ten identifiable classes to one costs nearly as much again as reaching ten. FRONT spends $\max_i\hat B_i=1.201$, nearly twice $\hat B^*_{\max}$, yet step~(iii) returns $|\calA^{\mathrm{pc}}|=7$, missing every target below $k_0=7$.

\textbf{Tightening.} Sec.~\ref{sec:gap} bounds only one of the two steps from below, so the remedy with a guaranteed target is the single-statistic step---aggregating directions, by orthogonal $1$-Lipschitz projections or a sliced-$W_1$ construction. Whether a sharper divergence inequality helps is unsettled, though one candidate is ruled out: swapping Pinsker for Bretagnolle--Huber \cite{tsybakov2009}, $D_{\mathrm{KL}}\ge-\log_2(1-\TV^2)$ in bits, would make matters \emph{worse}: the two cross at $\TV\approx0.893$, and the bound feeds Pinsker the Lemma~\ref{lem:drift} floor, never above $0.474$ ($0.423$ same-space), well inside where Pinsker is stronger.

\textbf{Feature dimension and deep attacks.} Theorem~\ref{thm:main} constrains any observable reached through a Lipschitz map, deep fingerprinting embeddings \cite{sirinam2018deep,deng2025countmamba} included. The measured side is monotone in the observable; the bound is not, since it scales as $(\dovr)^2/M^2$ and a richer space enlarges both---so the $2.8\times$ gained here by the 11-D estimate is specific to this data. What the proxy cannot do is predict a deep attack's absolute accuracy; estimating $W_1$ in an embedding space meets the curse of dimensionality and is left open.

\textbf{Boundary of the conditions.} A single $C$ applies to all classes; with heterogeneous $C_i$, $C=\max_iC_i$ is conservative. Since $\sup_xW_1\le B$ implies $\E_X[W_1]\le B$ but not conversely, a bound over $\calD(B)$ does \emph{not} extend to the family cut out by a global constraint, where one class may be defended far beyond $B$; Corollary~\ref{cor:perclass} on realized costs covers that, no a fortiori argument does. Two earlier claims deserve confronting. The ``centroid class''---$\dovr=0$, inactive at any budget---does not occur here: the smallest $\hat{\bar\Delta}_i^{\mathrm{ovr}}$ is $61\%$ of the mean in $Y^{(11)}$ and $52\%$ in $Y^{(3)}$; a real gap, but unexercised. And $\Theta(1)$ concerns a limit that $k=95$ and $k=100$ cannot resolve; we verify the summation structure, not its asymptotics. Optimal budget allocation across classes, which FRUGAL \cite{wang2026frugal} approaches heuristically, is also open.

\section{Conclusion}\label{sec:conclusion}

We extended side-channel leakage inevitability from the binary undefended setting to $k$-class defended traffic analysis, through a summation-form bound, kept nonvacuous by a per-class budget corollary where the uniform-budget bound is empty, a cascade of per-class critical costs, and an accuracy corollary with a prior-free per-class counterpart.

Empirically the summation form is $1.45\times$ stronger than the best pairwise construction, a modest margin: what the decomposition buys is the per-class certificate, not the magnitude. The measured MI has a strictly positive $95\%$ confidence lower bound under every defense, undefended OVR distinguishability predicts post-defense leakage at Spearman $\rho=0.62$--$0.77$, and the framework transfers unchanged across the QUIC and TCP stacks. The bound stays conservative, but accountably so: the ratio factors into class exclusion, the inequality chain, and the estimation space, and matching the index set alone reduces FRONT's $122\times$ gap to $21\times$ against $15\times$ undefended. The limit of the sufficient condition is visible too: under TrafficSliver no class is active while $0.482$ bit is still measurable.

For system design the theory supplies two attack-model-free metrics: $k^{\mathrm{pc}}_{\mathrm{eff}}$, how many classes a deployed defense still leaves certifiably identifiable, and $B^*_{\max}$, a budget floor rather than a sufficient spend---FRONT exceeds it yet seven classes remain above the per-class criterion. Aggregating Lipschitz directions to tighten Lemma~\ref{lem:drift}, allocating a budget across classes, and pairing this converse with the achievability side \cite{liu2026ratedistortion,wang2026frugal} into a sandwich characterization are natural next steps.

\section*{Acknowledgment}
This work was supported by the National Natural Science Foundation of China Joint Fund Integration Project (No.~U2436601). We thank the anonymous reviewers, whose request to explain rather than report the theory--measurement gap led directly to Sec.~\ref{sec:gap}.


\begin{thebibliography}{00}

\bibitem{liu2026inevitability}
G.~Liu, G.~Cheng, and W.~Liu, ``The inevitability of side-channel leakage in encrypted traffic,'' \textit{Acta Electronica Sinica}, vol.~54, no.~2, pp.~837--850, 2026.

\bibitem{liu2026ratedistortion}
G.~Liu, G.~Cheng, W.~Liu, and Y.~Wang, ``Rate-distortion function for encrypted traffic side-channel defense,'' \textit{Sci. Sin. Inform.}, 2026, in press, doi:10.1360/SSI-2026-0081.

\bibitem{sirinam2018deep}
P.~Sirinam, M.~Imani, M.~Juarez, and M.~Wright, ``Deep fingerprinting: undermining website fingerprinting defenses with deep learning,'' in \textit{Proc. ACM CCS}, 2018, pp.~1928--1943.

\bibitem{zhao2024finegrained}
X.~Zhao, X.~Deng, Q.~Li, \textit{et al.}, ``Towards fine-grained webpage fingerprinting at scale,'' in \textit{Proc. ACM CCS}, 2024.

\bibitem{deng2025countmamba}
X.~Deng, R.~Zhao, Y.~Wang, \textit{et al.}, ``CountMamba: a generalized website fingerprinting attack via coarse-grained representation and fine-grained prediction,'' in \textit{Proc. IEEE S\&P}, 2025, pp.~1419--1437.

\bibitem{shen2023machine}
M.~Shen, K.~Ye, X.~Liu, \textit{et al.}, ``Machine learning-powered encrypted network traffic analysis: a comprehensive survey,'' \textit{IEEE Commun. Surveys Tuts.}, vol.~25, no.~1, pp.~791--824, 2023.

\bibitem{hasselquist2024raisingbar}
D.~Hasselquist, E.~Witwer, A.~Carlson, \textit{et al.}, ``Raising the bar: improved fingerprinting attacks and defenses for video streaming traffic,'' \textit{Proc. Privacy Enhancing Technol.}, vol.~2024, no.~4, pp.~167--184, 2024.

\bibitem{shen2024palette}
M.~Shen, K.~Ji, J.~Wu, \textit{et al.}, ``Real-time website fingerprinting defense via traffic cluster anonymization,'' in \textit{Proc. IEEE S\&P}, 2024, pp.~3238--3256.

\bibitem{wang2026frugal}
R.~Wang, Z.~Ling, G.~Liu, \textit{et al.}, ``Cease at the ultimate goodness: towards efficient website fingerprinting defense via iterative mutual information minimization,'' in \textit{Proc. NDSS}, 2026.

\bibitem{cover2006elements}
T.~M.~Cover and J.~A.~Thomas, \textit{Elements of Information Theory}, 2nd~ed.\hskip 1em plus 0.5em minus 0.4em Wiley, 2006.

\bibitem{li2018measuring}
S.~Li, H.~Guo, and N.~Hopper, ``Measuring information leakage in website fingerprinting attacks and defenses,'' in \textit{Proc. ACM CCS}, 2018, pp.~1977--1992.

\bibitem{cherubin2017bayes}
G.~Cherubin, ``Bayes, not na\"ive: security bounds on website fingerprinting defenses,'' \textit{Proc. Privacy Enhancing Technol.}, vol.~2017, no.~4, pp.~215--231, 2017.

\bibitem{gong2020front}
J.~Gong and T.~Wang, ``Zero-delay lightweight defenses against website fingerprinting,'' in \textit{Proc. USENIX Security}, 2020, pp.~717--734.

\bibitem{juarez2016wtfpad}
M.~Juarez, M.~Imani, M.~Perry, C.~Diaz, and M.~Wright, ``Toward an efficient website fingerprinting defense,'' in \textit{Proc. ESORICS}, LNCS 9878, 2016, pp.~27--46.

\bibitem{cadena2020trafficsliver}
W.~De~la~Cadena, A.~Mitseva, J.~Pennekamp, \textit{et al.}, ``TrafficSliver: fighting website fingerprinting attacks with traffic splitting,'' in \textit{Proc. ACM CCS}, 2020.

\bibitem{wang2017walkietalkie}
T.~Wang and I.~Goldberg, ``Walkie-Talkie: an efficient defense against passive website fingerprinting attacks,'' in \textit{Proc. USENIX Security}, 2017, pp.~1375--1390.

\bibitem{smith2021quic}
J.-P.~Smith, P.~Mittal, and A.~Perrig, ``Website fingerprinting in the age of QUIC,'' \textit{Proc. Privacy Enhancing Technol.}, vol.~2021, no.~2, pp.~48--69, 2021.

\bibitem{tsybakov2009}
A.~B.~Tsybakov, \textit{Introduction to Nonparametric Estimation}.\hskip 1em plus 0.5em minus 0.4em Springer, 2009.

\end{thebibliography}
\end{document}